\documentclass[conference,9pt]{IEEEtran}
\IEEEoverridecommandlockouts
\usepackage{cite}
\usepackage{amsmath,amssymb,amsfonts,amsthm}
\usepackage{algorithmic}
\usepackage{graphicx}
\usepackage{textcomp}
\usepackage{xcolor}
\def\BibTeX{{\rm B\kern-.05em{\sc i\kern-.025em b}\kern-.08em
    T\kern-.1667em\lower.7ex\hbox{E}\kern-.125emX}}
\usepackage[
 labelformat=simple,  
]{subcaption}  
\usepackage[ruled, vlined]{algorithm2e} 
\usepackage{stfloats}
\usepackage[utf8]{inputenc} 
\usepackage{booktabs}
\usepackage{calc}
\usepackage{siunitx}
\usepackage{mathtools}
\usepackage{tikz}
\usepackage{tkz-fct}
\usepackage{tkz-graph}
\usepackage{circuitikz}
\usepackage{adjustbox}
\usepackage{xstring}
\usepackage[
 capitalise,
 nameinlink,
 noabbrev]{cleveref}

\usetikzlibrary{shapes}
\usetikzlibrary{matrix, arrows,fit,shapes.gates.logic.US,shapes.gates.logic.IEC,calc,backgrounds,decorations.pathmorphing}
\usetikzlibrary{calc}

\newcommand{\AND}{\textsc{And}}
\newcommand{\OR}{\textsc{Or}}

\newcommand{\NOR}{\textsc{Nor}}
\newcommand{\INV}{\textsc{Inv}}

\newcommand{\OAI}{\textsc{Oai}}
\newcommand{\AOP}{\textsc{And}-\textsc{Or} Path}
\newcommand{\aop}{\textsc{And}-\textsc{Or} path}
\newcommand{\BL}{LO}
\newcommand{\huffmancircuit}{undetermined circuit}
\newcommand{\huffmancircuits}{undetermined circuits}
\newcommand{\Huffmancircuit}{Undetermined circuit}
\DeclareMathOperator{\argmin}{argmin}

\DeclareMathOperator{\generalout}{out}
\DeclareMathOperator{\out}{\generalout(C)}
\newcommand{\dgate}{d_{\text{gate}}}
\newcommand{\ddist}{d_{\text{dist}}}

\DeclareMathOperator{\N}{\mathbb N}
\DeclareMathOperator{\R}{\mathbb R}

\newcommand{\deltamin}{\delta_{\text{min}}}
\newcommand{\deltarangeit}{\delta_{\text{it}}}
\newcommand{\targetdelta}{\delta_{\text{target}}}
\newcommand{\itrange}{\text{num}_{\text{it}}}
\newcommand{\maxpathlength}{m_{\text{max}}}
\definecolor{darkgreen}{RGB}{0, 100, 0}
\definecolor{darkblue}{RGB}{0, 0, 100}
\colorlet{atcolor}{black}
\colorlet{and_color}{cyan!50}
\colorlet{or_color}{yellow!50}
\newcommand{\rmColorBtr}{darkgreen}

\newcommand{\rmColorEql}{blue}
\newcommand{\rmColorUkn}{black}
\newcommand{\rmDecorate}[2][]{%
 \ifthenelse{\equal{#1}{p}}{%
   \IfDecimal{#2}{\SI[retain-explicit-plus, group-minimum-digits=3]{#2}{\percent}}{#2}%
 }{%
   \IfDecimal{#2}{\num[retain-explicit-plus, group-minimum-digits=3]{#2}}{#2}%
 }
}

\newcommand{\rmBtr}[2][]{\textcolor{\rmColorBtr}{\rmDecorate[#1]{#2}}}

\newcommand{\rmEql}[2][]{\textcolor{\rmColorEql}{\rmDecorate[#1]{#2}}}
\newcommand{\rmUkn}[2][]{\textcolor{\rmColorUkn}{\rmDecorate[#1]{#2}}}
\newcommand{\chipName}[1]{#1}

\SetCommentSty{mycommfont}

\newcommand{\showid}[2]{#2}
\newcommand{\showwiredelay}[1]{}

\newcommand{\boxedDescription}[5]{%
 {%
  \setlength{\parskip}{0.6\topsep}
   {%
    \setlength{\fboxsep}{5pt}
     {%
      \setlength{\parindent}{0pt}
       \framebox[\columnwidth][c]{%
        \begin{tabular}{%
          @{\hspace{\fboxsep}}
          l@{}
          @{\hspace{\fboxsep}}
          p{\columnwidth-3\fboxsep-\maxof{\widthof{#1}}{\widthof{#2}}}@{\hspace{\fboxsep}}
        }%
         \multicolumn{2}{p{\columnwidth-2\fboxsep}}{%
          {\Large{\textsc{#3}}}} \\ %
         \textit{#1:} & #4 \\ %
         \textit{#2:} & #5 \\ %
        \end{tabular}%
       }%
     }%
   }%
 }%
}

\newcommand{\problemBase}[3]{%
 \boxedDescription{Instance}{Task}{#1}{#2}{#3}%
}

\newcommand{\problemIndex}[4]{
 \index{#2|(}%
 \problemBase{#1}{#3}{#4}%
 \index{#2|)}%
}

\newcommand{\problem}[4]{
 \problemIndex{#1}{#1@\textsc{#1}}{#2}{#3}
 \expandafter\newcommand\csname #4\endcsname{\textsc{#1}~problem}%
}

\newtheorem{theorem}{Theorem}
\newtheorem{lemma}[theorem]{Lemma}
\theoremstyle{definition}
\newtheorem{definition}[theorem]{Definition}

\begin{document}

\title{Delay Optimization of Combinational Logic \\ by \AOP{} Restructuring}

\author{\IEEEauthorblockN{Ulrich Brenner and Anna Hermann}
\IEEEauthorblockA{\textit{Research Institute for Discrete Mathematics, University of Bonn} \\
\{brenner,hermann\}@dm.uni-bonn.de}
}

\maketitle

\addtolength{\textfloatsep}{-15pt}
\addtolength{\dbltextfloatsep}{-5pt} 
\addtolength{\floatsep}{-5pt}

%
\begin{abstract}

We propose a dynamic programming algorithm
that constructs delay-optimized circuits for alternating \aop{}s
with prescribed input arrival times.
Our algorithm fulfills best-known approximation guarantees
and empirically outperforms earlier methods
by exploring a significantly larger portion of the solution space.

Our algorithm is the core of
a new timing optimization framework that replaces critical paths of arbitrary length by logically
equivalent realizations with less delay.
Our framework allows revising early decisions on the logical structure of the netlist
in a late step of an industrial physical design flow.
Experiments demonstrate the effectiveness
of our tool on 7nm real-world instances.

\end{abstract}

%
%
%

%
%
\maketitle

\section{Introduction}

In VLSI design, logic synthesis
turns the abstract logic specification of a chip
into a concrete representation in terms of gates.
This happens very early in the design process,
and for the following steps, the logical description typically remains fixed.
However, during physical design,
it may turn out that the chosen implementation of the logic
functionality was not the best choice, e.g., with respect to placement or timing.
Now it would be desirable to find a better suited logically equivalent representation.

We propose an algorithm that improves timing by
logic restructuring of critical combinational paths.
Optimizing a path boils down to optimizing an \aop{},
i.e., a Boolean function of type
${t_0 \land (t_1 \lor (t_2 \land (t_3 \lor (t_4 \land (  \dots t_{m-1}) \dots )}$,
see \cite{WerberEtal2007}.

Besides, \aop{}s have an important application in the construction
of adder circuits. The carry bit computation in an
adder (which is the critical part) is equivalent to the evaluation of
an \aop{}.
The tasks of \aop{} and adder optimization are actually equivalent
concerning timing
if circuit size is disregarded.

Many efficient adder circuits (e.g.,\ \cite{Brent1982,KoggeStone1973,Khrapchenko1970})
have been proposed in the previous
decades and could hence be used for optimizing \aop{}s.
In terms of depth, the best approximation guarantee
for \aop{} circuits has been proven by \cite{Grinchuk2009}.
However,
these approaches optimize circuit \emph{depth},
yielding fast circuits only if all input signals arrive simultaneously.
In our setting on the most timing-critical path, this will rarely be the case.
Instead, we minimize circuit \emph{delay},
a generalization of circuit depth that takes into account
individual prescribed input \emph{arrival times}.

Some algorithms for adder optimization regard
input arrival times, but most lack provable guarantees:
For adders with general arrival times, there are a greedy heuristic \cite{YehJen2003}
and a dynamic program \cite{Liuetal2003},
but for both, no approximation ratio can be shown.
In \cite{Royetal2014}, the delay of adders is evaluated
regarding arrival times computed after physical
design, but the optimization goal is depth and not delay.

Algorithms for \aop{} optimization with input arrival times
that achieve provably good approximation ratios
are presented in \cite{BrennerHermann2019}, \cite{RautenbachEtal2006} and \cite{HeldSpirkl2017}.
We will explain their ideas
in \cref{sec::prev_algorithms}.
The method of \cite{RautenbachEtal2006}
is used in \cite{WerberEtal2007} to optimize general logic paths.

Our goal is to restructure critical paths of any length
with provably good approximation guarantees.
In contrast, many other approaches synthesize whole netlists and thus arbitrary Boolean functions.
As in general, finding a logically equivalent implementation of a given circuit with, say, minimum depth
is an NP-hard problem,
these approaches only replace
sub-circuits of constant size by alternative realizations
(see e.g., \cite{AmaruEtal2017,Cortadella2003,MishchenkoEtal2011,Stoketal1996,Plazaetal2008}).
Here, the new solution is logically correct by construction,
but an extension to larger sub-circuits is hardly possible.

Our main contributions are:
\begin{itemize}
\item We propose a new dynamic program for delay optimization of \aop{}s.
   In fact, the algorithm solves a more general problem, the optimization
   of so-called extended \aop{}s.
   We describe how decisions on the structure of sub-solutions
   can be postponed until these sub-solutions are combined.
   This reduces rounding effects that are inherent
   in previous algorithms.
\item Our algorithm fulfills best known theoretical delay guarantees
   as it is a common generalization of all previously best approaches
   \cite{HeldSpirkl2017,BrennerHermann2019,RautenbachEtal2006}.
   Moreover, we demonstrate in experiments that
   we improve delay significantly compared to those.
\item We compute lower bounds on the best possible delay of \aop{}s.
   On 89\% of our test instances, the result of our algorithm
   matches the lower bound and is thus provably optimum.
\item We propose a framework for timing optimization of combinational paths of arbitrary length
   based on \cite{WerberEtal2007} with our \aop{} restructuring algorithm as a core routine.
   The generic delay model used in our core algorithm allows
   incorporating physical locations.
   Our framework contains several classical timing-optimization
   tools and
   -- in contrast to the simple mapping used in
   \cite{WerberEtal2007} --
   an evolved technology-mapping method~\cite{Elbert2017}.
\item Experiments on recent
   industrial 7nm chips show the efficiency and effectiveness of
   our framework. We improve worst slack and total slack considerably without any
   impact on other metrics.
\end{itemize}

The rest of the paper is organized as follows.
In \cref{sec:aop alg}, we define the \aop{} optimization problem,
survey known approaches, and present our new approximation algorithm.
\cref{sec:flow} describes our logic restructuring framework.
Experimental results are shown in \cref{sec:experiments},
and \cref{sec:conclusion} contains concluding remarks.

\section{And-Or Path Optimization} \label{sec:aop alg}

Note that in this section,
we use a simplified linear delay model
 with unit gate delay and zero wire delay.
In \cref{sec:normalization},
we will generalize this model
to adapt to our application in physical design.

\subsection{Problem formulation}

For us, a \emph{circuit} $C$ is a connected acyclic digraph
whose nodes can be partitioned into two sets:
\emph{inputs} with no incoming edges representing Boolean variables, and
\emph{gates} representing an elementary Boolean function (mostly \AND{}2 or \OR{}2,
i.e., \AND{} and \OR{} gates with fan-in two),
where only a single gate $\out$ called \emph{output} has no outgoing edges.
An \emph{\aop{}} on inputs $t_0,\dots,t_{m-1}$ is a Boolean formula of type
\begin{align*}
   g(t_0,\dots,t_{m-1}) &= t_0 \land (t_1 \lor (t_2 \land (t_3 \lor (t_4 \land (  \dots t_{m-1}) \dots ) \text{ or}\\
   g^*(t_0,\dots,t_{m-1}) &= t_0 \lor (t_1 \land (t_2 \lor (t_3 \land (t_4 \lor (  \dots t_{m-1}) \dots )\,.
\end{align*}

On the left-hand side of \cref{fig::arrival_time_computation_example},
a circuit for the \aop{} $g(t_0, t_1, t_2, t_3, t_4)$ is shown.
Given individual \emph{arrival times} $a(t_i) \in \R$ for each input signal $t_i$,
$i=0, \dots, m-1$, we ask for a Boolean circuit
computing $g(t_0,\dots,t_{m-1})$ that consists of \AND2{} and \OR2{} gates only
and is timing-wise best possible in the following sense:
We assume that traversing a gate takes $1$ time unit, so the
\emph{gate arrival time} is the maximum of its predecessors' arrival
times plus $1$. By scanning a circuit $C$ from the inputs to the output,
we can compute arrival times at all gates.
The \emph{delay} of a circuit is defined as the arrival time at $\out$.
Summarizing, we study the following problem:

\problem{And-Or Path Optimization}
        {$m \in \N$,
         Boolean input variables $t = (t_0, \dotsc, t_{m-1})$,
         arrival times $a(t_0), \dotsc, a(t_{m-1}) \in \R$.}
        {Compute a circuit $C$ using only \AND{}2 and \OR{}2 gates realizing $g(t)$ or $g^*(t)$ with minimum possible delay.}
        {praopdelayopt}

        Figure~\ref{fig::arrival_time_computation_example} shows how
gate arrival times are computed in two circuits that
both realize the \aop{} $g(t_0, t_1, t_2, t_3, t_4)$.
The circuits have a delay of 7 and 6, respectively.
Note that in the special case when all input arrival times are $0$,
circuit delay is exactly circuit depth, i.e., the length of a longest
directed path.

\begin{figure}

\newcommand{\scalefactor}{1.2}
\newcommand{\subfigwidth}{0.36\columnwidth}
\newcommand{\picwidth}{0.97\columnwidth}

\centering
\adjustbox{valign=t}{
\begin{subfigure}[]{0.34\columnwidth}

\resizebox{\picwidth}{!}{%
\begin{tikzpicture}

\node[ outer sep=0pt, or gate US, fill=or_color, draw, logic gate inputs=nn, rotate=270, thick] at (5,2) (or1){\rotatebox{90}{\color{atcolor}$4$}};
\node[ outer sep=0pt, and gate US, fill=and_color, draw, logic gate inputs=nn, rotate=270, thick] at (4,1) (and1){\rotatebox{90}{\color{atcolor}$5$}};

\node[ outer sep=0pt, or gate US, fill=or_color, draw, logic gate inputs=nn, rotate=270, thick] at (3,0) (or4){\rotatebox{90}{\color{atcolor}$6$}};
\node[ outer sep=0pt, and gate US, fill=and_color, draw, logic gate inputs=nn, rotate=270, thick] at (2,-1) (and2){\rotatebox{90}{\color{atcolor}$7$}};

\node[outer sep=0pt, atcolor, scale=\scalefactor] (a5) at (1.5, 3.5){$2$};
\node[outer sep=0pt, atcolor, scale=\scalefactor] (a4) at (2.5, 3.5){$2$};
\node[outer sep=0pt, atcolor, scale=\scalefactor] (a3) at (3.5, 3.5){$3$};
\node[outer sep=0pt, atcolor, scale=\scalefactor] (a2) at (4.5, 3.5){$1$};
\node[outer sep=0pt, atcolor, scale=\scalefactor] (a1) at (5.5, 3.5){$3$};

\node[outer sep=0pt, scale=\scalefactor] (i5) at (1.5, 3.){$t_0$};
\node[outer sep=0pt, scale=\scalefactor] (i4) at (2.5, 3.){$t_1$};
\node[outer sep=0pt, scale=\scalefactor] (i3) at (3.5, 3.){$t_2$};
\node[outer sep=0pt, scale=\scalefactor] (i2) at (4.5, 3.){$t_3$};
\node[outer sep=0pt, scale=\scalefactor] (i1) at (5.5, 3.){$t_4$};

\draw[thick] (or1.output) -- (and1.input 1);
\draw[thick] (i2) -- (or1.input 2);
\draw[thick] (i3) -- (and1.input 2);
\draw[thick] (i1) -- (or1.input 1);

\draw[thick] (or4.output) -- (and2.input 1);

\draw[thick] (i5) -- (and2.input 2);

\draw[thick] (i4) -- (or4.input 2);
\draw[thick] (and1.output) -- (or4.input 1);

\end{tikzpicture}
}
\end{subfigure}
}
\adjustbox{valign=t}{
\begin{subfigure}[]{\subfigwidth}

\resizebox{\picwidth}{!}{%
\begin{tikzpicture}

\node[outer sep=0pt, atcolor, scale=\scalefactor] (a5) at (1.5, 3.5){$2$};
\node[outer sep=0pt, atcolor, scale=\scalefactor] (a4) at (2.5, 3.5){$2$};
\node[outer sep=0pt, atcolor, scale=\scalefactor] (a3) at (3.5, 3.5){$3$};
\node[outer sep=0pt, atcolor, scale=\scalefactor] (a2) at (4.5, 3.5){$1$};
\node[outer sep=0pt, atcolor, scale=\scalefactor] (a1) at (5.5, 3.5){$3$};

\node[outer sep=0pt, scale=\scalefactor] (i5) at (1.5, 3.){$t_0$};
\node[outer sep=0pt, scale=\scalefactor] (i4) at (2.5, 3.){$t_1$};
\node[outer sep=0pt, scale=\scalefactor] (i3) at (3.5, 3.){$t_2$};
\node[outer sep=0pt, scale=\scalefactor] (i2) at (4.5, 3.){$t_3$};
\node[outer sep=0pt, scale=\scalefactor] (i1) at (5.5, 3.){$t_4$};

\node[ outer sep=0pt, or gate US, fill=or_color, draw, logic gate inputs=nn, rotate=270, thick] at (3,2) (or1){\rotatebox{90}{\color{atcolor} $4$}};
\node[ outer sep=0pt, or gate US, fill=or_color, draw, logic gate inputs=nn, rotate=270, thick] at (4,2) (or2){\rotatebox{90}{\color{atcolor}$3$}};
\node[ outer sep=0pt, and gate US, fill=and_color, draw, logic gate inputs=nn, rotate=270, thick] at (2.5,1) (and1){\rotatebox{90}{\color{atcolor}$5$}};
\node[ outer sep=0pt, or gate US, fill=or_color, draw, logic gate inputs=nn, rotate=270, thick] at (4.5,1) (or3){\rotatebox{90}{\color{atcolor}$4$}};
\node[ outer sep=0pt, and gate US, fill=and_color, draw, logic gate inputs=nn, rotate=270, thick] at (3.5,0) (and2){\rotatebox{90}{\color{atcolor}$6$}};

\draw[thick] (i3) -- (or1.input 1);
\draw[thick] (i4) -- (or1.input 2);
\draw[thick] (i2) -- (or2.input 1);
\draw[thick] (i4) -- (or2.input 2);
\draw[thick] (or1.output) -- (and1.input 1);
\draw[thick] (i5) -- (and1.input 2);
\draw[thick] (or2.output)   --  (or3.input 2);
\draw[thick] (i1) -- (or3.input 1);
\draw[thick] (or3.output) -- (and2.input 1);
\draw[thick] (and1.output)   -- (and2.input 2);


\node[outer sep=0pt, atcolor] (z) at (2, -1.3){};

\end{tikzpicture}
}

\end{subfigure}
}

\caption{Two circuits computing the function $g(t_0, \dotsc, t_4) = t_0 \land (t_1 \lor (t_2 \land (t_3 \lor t_4)))$
with input and gate arrival times.}
\label{fig::arrival_time_computation_example}
\end{figure}

Given an instance consisting of inputs $t_0, \dotsc, t_{m-1}$
with arrival times $a(t_0), \dots, a(t_{m-1}),$ we define
the \emph{weight}
$
    W := \sum_{i=0}^{m-1} 2^{a(t_i)}
$.
It is not too difficult to see that
$\lceil\log_2(W)\rceil$ is a lower bound for the delay of any binary
circuit computing an \aop{} for inputs $t_0,\dots,t_{m-1}$
with arrival times $a(t_0),\dots,a(t_{m-1}) \in \N$
(this boils down to Kraft's inequality \cite{Kraft1949};
see \cite{RautenbachEtal2006} for a concise proof).

Optimizing $g(t)$ and $g^*(t)$
is equivalently hard:
By the duality principle of Boolean algebra,
any circuit for $g(t)$ consisting of \AND{} and \OR{} gates
can be transformed into a circuit for $g^*(t)$ with the
same delay by exchanging \AND{} and \OR{} gates and vice versa.

\subsection{Previous Algorithms}\label{sec::prev_algorithms}

A common approach for \aop{} optimization is the application of
\emph{recursion formulas} that allow reducing the problem
to the construction of circuits for \aop{}s with fewer inputs.

The algorithm by Rautenbach et. al \cite{RautenbachEtal2006} is based on the
following equation (for $\lambda \in \N$ with $2 \lambda < m - 2$):
\begin{align}\label{eq::rec_rautenbach_primal}
   g(t_0,\dots,t_{m-1}) & =  g(t_{0},\dots,t_{2\lambda-1}) \\
                       & \lor   \big( t_0 \land t_2 \land t_4 \land \dots \land t_{2\lambda-2}
                             \land g(t_{2\lambda},\dots, t_{m-1}) \big) \nonumber
\end{align}
To see the correctness of \labelcref{eq::rec_rautenbach_primal},
check that $g(t_0,\dots,t_{m-1})$ is true
exactly in the following two cases:
\begin{itemize}
\item $g(t_0,\dots,t_{2\lambda-1})$ is true (then the
  other inputs do not matter)
\item $g(t_{2\lambda}, \dots, t_{m-1})$ is true and the value ``true''
  is propagated to the output because the inputs
  $t_0, t_2,  t_4, \dots, t_{2\lambda-2}$ are all true
\end{itemize}
See \cite{RautenbachEtal2006} for a detailed proof.
Using formula~(\ref{eq::rec_rautenbach_primal}),
an \aop{} circuit on inputs ${t_0,\dots,t_{m-1}}$ can be
constructed by combining \aop{} circuits on inputs
$t_{0}, \dots, t_{2\lambda-1}$ and on inputs
$t_{2\lambda}, \dots, t_{m-1}$ and a circuit for a multi-input {\sc And}
on the inputs $t_0, t_2, \dotsc, t_{2\lambda-2}$.
Using \labelcref{eq::rec_rautenbach_primal} in a dynamic program with running time $\mathcal O(m^3)$,
the authors of \cite{RautenbachEtal2006}
construct \aop{} circuits with delay at most
$1.441\log_2 (W) + 3$.
Held and Spirkl \cite{HeldSpirkl2017} obtain a slightly better delay bound
of $1.441\log_2 (W) + 2.673$
using the dual of the following equation (for $\lambda$ with $2\lambda < m - 1$):
\begin{align} \label{eq::rec_spirkl}
g(t_0, \dotsc, t_{m-1}) &= g(t_0, \dotsc, t_{2\lambda}) \\
                          & \land \big((t_1 \lor t_3 \lor \dotsc \lor t_{2\lambda + 1})
                             \lor g(t_{2\lambda+2}, \dotsc, t_{m-1})\big) \nonumber
\end{align}
Their algorithm runs in time $\mathcal O(m \log_2^2 m)$
as they explicitly choose $\lambda$ in each recursion step.
The proof of \labelcref{eq::rec_spirkl}
is analogous to the proof of \labelcref{eq::rec_rautenbach_primal},
but here one should check in which cases the two formulas are false.
We will use \labelcref{eq::rec_spirkl} in a slightly different equivalent form
(note that $t_{2\lambda + 1} \lor g(t_{2\lambda+2}, \dotsc, t_{m-1}) =
g^*(t_{2\lambda+1}, \dotsc, t_{m-1})$):
\begin{align} \label{eq::rec_spirkl_shift}
g(t_0, \dotsc, t_{m-1}) &= g(t_0, \dotsc, t_{2\lambda}) \\
                          & \land \big((t_1 \lor t_3 \lor \dotsc \lor t_{2\lambda - 1})
                             \lor g^*(t_{2\lambda+1}, \dotsc, t_{m-1})\big) \nonumber
\end{align}

As \labelcref{eq::rec_rautenbach_primal,eq::rec_spirkl_shift}
contain functions combining a multi-input \AND{} or \OR{} with an \aop{},
we define for
$t = (t_0, \dotsc, t_{m-1})$ and $0 \leq i \leq j \leq k < m$ with $j - i$ even
the \emph{extended \aop{}s}
\begin{align*}
\phi_{i, j, k}    &:= t_i \land t_{i + 2} \land \dotsc \land t_{j-4} \land t_{j - 2} \land g(t_j, \dotsc, t_k) \quad \text{ and}\\
\phi^*_{i, j, k}  &:= t_i \lor  t_{i + 2} \lor \dotsc \lor  t_{j-4} \lor  t_{j - 2} \lor  g^*(t_j, \dotsc, t_k)\,.
\end{align*}
The extended \aop{} $\phi_{0, 4, 12}$ is depicted in \cref{ext_instance}.
From the splits \labelcref{eq::rec_rautenbach_primal,eq::rec_spirkl_shift},
using extended \aop{}s as a more flexible replacement for sub-functions,
we deduce the splits
\begin{alignat}{5}
\phi_{0, 0, m-1} &= \phi_{0, 0, 2\lambda-1}
             &&\lor \phi_{0, 2\lambda, m-1}  && \text{ for } {\scriptstyle 1 \leq \lambda \leq \frac{m - 1}{2} \label{eq::rec_3_aop}} \\
\phi_{0, 0, m-1} &= \phi_{0, 0, 2\lambda}
             &&\land \phi^*_{1, 2\lambda+1, m-1} && \text{ for } {\scriptstyle  0 \leq \lambda \leq \frac{m - 2}{2}} \label{eq::rec_2_aop}
\shortintertext{that can be generalized to extended \aop{}s as in}
\phi_{i, j, k} &= \phi_{i, j, j + 2\lambda - 1}
             &&\lor \phi_{i,j+2\lambda, k} && \text{ for } {\scriptstyle 1 \leq \lambda \leq \frac{k - j}{2}\,, \label{eq::rec_3}} \\
\phi_{i, j, k} &= \phi_{i, j ,j + 2\lambda}
             &&\land \phi^*_{j + 1, j + 2\lambda + 1, k} && \text{ for } {\scriptstyle 0 \leq \lambda \leq \frac{k - j - 1}{2}}\,. \label{eq::rec_2}
\end{alignat}

Note that in \labelcref{eq::rec_3} and \labelcref{eq::rec_2},
the functions on the right-hand side depend on fewer inputs than $\phi_{i, j, k}$.
\cref{fig::arrival_time_computation_example} shows an example
for split~\labelcref{eq::rec_2_aop} with $\lambda = 1$,
and \cref{ext_instance,ext_split} for split~\labelcref{eq::rec_2} with $\lambda = 2$.

Using split~\labelcref{eq::rec_2}
and its dual, Grinchuk~\cite{Grinchuk2009}
proves the upper bound $\log_2 m + \log_2 \log_ 2 m + 3$ on the depth of \aop{} circuits,
and Brenner and Hermann \cite{BrennerHermann2019} give an algorithm for arbitrary integer arrival times
with running time $\mathcal O(m^2 \log_2 m)$ and a delay bound of
\begin{equation} \label{our bound}
   \log_2 W + \log_2 \log_2 m + \log_2 \log_2 \log_2 m + 4.3.
\end{equation}

In the special case when $k - j \leq 1$,
$\phi_{i, j, k}$ is actually a multi-input \AND{},
and the function can be realized by a delay-optimum circuit
using a greedy algorithm called \emph{Huffman coding}:

\begin{theorem}[Golumbic \cite{Golumbic1976}, based on Huffman \cite{Huffman1952}]\label{theorem::huffman}
Given inputs $t_0, \dotsc, t_{m-1}$ with arrival times $a(t_i)$,
a delay-optimum circuit for the Boolean function $t_0 \land \dotsc \land t_{m-1}$
(or $t_0 \lor \dotsc \lor t_{m-1}$)
can be constructed in $\mathcal O(m \log_2 m)$ time.
If $a(t_i) \in \N$ for all ${i = 0, \dotsc, m-1}$, then
the delay of an optimum circuit is $\lceil \log_2 (W) \rceil$.
\end{theorem}




%

\subsection{Our Approach} \label{sec: aop alg generalization}

We present an algorithm for \aop{} optimization with prescribed input arrival times
that generalizes any of
the algorithms in \cite{BrennerHermann2019,HeldSpirkl2017,RautenbachEtal2006}.
In particular, on any instance, the delay of our solution is at least as good as the
delay computed with any of the three algorithms,
and on most instances, it is better, cf. \cref{delay-comp}.

To simplify notations, hereafter we assume that all arrival times are
integral. Still, our implementation allows arbitrary arrival times.


Recall that in the \aop{} optimization problem,
we aim at computing a circuit containing only fan-in-$2$ gates.
However, in intermediate steps, we allow a larger
fan-in for the gate computing the output of the circuit.
This leads to the following definition.

\begin{definition}
An {\em \huffmancircuit{}} is a Boolean circuit $C$ consisting of
\AND{} and \OR{} gates only such that all gates with the possible exception
of $\out$ have fan-in two. With given input arrival times,
the {\em weight} of $C$ is
$\text{weight}(C) := \sum_{i=1}^k 2^{d_i}$, where $d_1, \dots, d_k$ are the arrival times at the
predecessors of $\out$.
\end{definition}

In Figure~\ref{fig::arrival_time_computation_example},
the weight of the left and right \huffmancircuit{} is $2^{2} + 2^{6} = 68$ and
$2^5 + 2^4 = 48$, respectively.
Figure~\ref{ext_impl} displays an \huffmancircuit{}
with fan-in $5$ at the output gate.

For \huffmancircuit{}s, we do not yet specify how we realize the output gate
by fan-in-2 gates. This allows greater flexibility when combining
several such circuits to a larger circuit.
The following \lcnamecref{lem-soft-huffman} shows that optimizing the weight of an \huffmancircuit{}
can be used to compute fan-in-2 circuits with small delay.

\begin{lemma} \label{lem-soft-huffman}
Given an \huffmancircuit{} $C$,
we can construct a
Boolean circuit using \AND2{} and \OR2{} gates only that
computes the same Boolean function as $C$ with delay at most
$\lceil \log_2(\text{weight}(C)) \rceil$.
\end{lemma}

\begin{proof}
Apply Huffman coding with the predecessors of $\out$ as inputs (see \cref{theorem::huffman}).
\end{proof}

\cref{alg::main_algorithm} states our overall dynamic programming algorithm
for \aop{} optimization on inputs $t_0, \dotsc, t_{m-1}$,
which works as follows:
We compute a cubic-size table that contains \huffmancircuits{}
$A_{i, j, k}$ and $O_{i, j, k}$ realizing the extended \aop{} $\phi_{i, j, k}$
for all ${0 \leq i \leq j \leq k \leq m-1}$ and $j-i$ even,
where $\generalout(A_{i, j, k}) = \AND{}$ and $\generalout(O_{i, j, k}) = \OR{}$.
In particular, this computes circuits for the entire \aop{}
$\phi_{0, 0, m-1} = g(t_{0},\dots,t_{m-1})$.

Note that when $k = j$ or $k = j+1$,
the function $\phi_{i, j, k}$ is a multiple-input \AND{},
hence, in \cref{line::huffman_coding},
an optimum solution can be found by Huffman coding (see \cref{theorem::huffman}).
To compute \huffmancircuit{}s for
$\phi_{i, j, k}$ with $j-i$ even and $k > j + 1$,
we assume that we have already computed \huffmancircuits{} for $\phi$
for instances with fewer inputs.
Then, in \cref{line::huff::def C},
we can enumerate all possible choices of $\lambda$ in the splits \labelcref{eq::rec_2,eq::rec_3}
to recursively compute a circuit $C$ for $\phi_{i, j, k}$ from pre-computed solutions
(while dualizing one sub-circuit accordingly in split \labelcref{eq::rec_2}).
Since the combination of two \huffmancircuits{} is not necessarily an \huffmancircuit{},
we apply Algorithm~\ref{alg::merge_circuits}.
Here, in \cref{merge_alg::huff},
we fix the structure of the undetermined sub-circuit $C_i$
as a circuit $C_i'$ over $\{\AND2, \OR2\}$.
\cref{fig:ext} shows an example of split \labelcref{eq::rec_2}.
In \cref{alg::main_algorithm}, the circuit $C$ is stored in a candidate list $\mathcal C$ of
\huffmancircuit{}s for $\phi_{i, j, k}$.
The \huffmancircuit{}s among $\mathcal C$ with the best weight with an \AND{} or \OR{} gate at the output
are stored as $A_{i, j, k}$ in \cref{line::compute_phi_start}
and $O_{i, j, k}$ in \cref{line::compute_phi_end}, respectively.

As final circuit for $\phi_{0, 0, m-1}$,
we choose the weight-minimum circuit among $A_{0, 0, m-1}$ and $O_{0, 0, m-1}$ in \cref{line::final_choice},
made a circuit over $\{\AND{}2, \OR{}2\}$ by \cref{lem-soft-huffman}.

\input{ext_aop_split.tex}

\newcommand{\Ceta}{37}

\LinesNumbered
 \begin{algorithm}
   \DontPrintSemicolon
    \KwIn{\Huffmancircuit{}s $C_1$ and $C_2$ computing
       Boolean functions $h_1$ and $h_2$; a gate type $\circ \in \{\AND, \OR\}$.}
    \KwOut{An \huffmancircuit{} $C$ computing $h_1 \circ h_2$.}
    \BlankLine
    Add a $\circ$ gate $c_0$ to the union of the circuits $C_1$ and $C_2$.\;
    \For{$i\gets1$ \KwTo $2$}
    {
      Let $c_1,\dots, c_k$ be the predecessors of $\generalout{}(C_i)$.\;
      \If{$\generalout{}(C_i)$ \text{\em is a $\circ$ gate}}
      {
         Remove $\generalout{}(C_i)$ and add edges $(c_1, c_0),\dots, (c_k, c_0)$.\;
      }
      \Else
      {
         Use \cref{lem-soft-huffman} to construct a circuit $C_i'$ from $C_i$.\; \label{merge_alg::huff}
         Add an edge from $\generalout{}(C_i')$ to $c_0$.\;
      }
    }
   \caption{Merging $2$ \huffmancircuits{}.}
   \label{alg::merge_circuits}
 \end{algorithm}

\SetKw{KwBy}{by}
 \begin{algorithm}
   \DontPrintSemicolon
    \KwIn{Boolean variables $t_0,\dots,t_{m-1}$ with arrival times
       $a(t_0),\dots,a(t_{m-1}) \in \N$.}
    \KwOut{A Boolean circuit computing $g(t_0,\dots,t_{m-1})$.}
    \BlankLine

    \For{$l\gets1$ \KwTo $m$}
    {\label{line::mainloop}
      \For{\small $0 \leq i \leq j \leq k < m$, $j - i$ even s.t.~$\phi_{i, j, k}$ has $l$ inputs}
      {
      \If(\tcp*[f]{$\phi_{i, j, k}$ multi-input \AND{}}){$k \in \{j, j + 1\}$}
      {
         $A_{i, j, k} := $ circuit computed by Huffman coding. \label{line::huffman_coding}
      }
      \Else
      {
         \label{line::init2}
            $\mathcal C :=$ list of \huffmancircuit{}s for $\phi_{i, j, k}$
            arising from applying split \labelcref{eq::rec_3} or \labelcref{eq::rec_2} with any valid $\lambda$,
            followed by a call to \cref{alg::merge_circuits}.\; \label{line::huff::def C}
            ${A_{i, j, k} := \argmin \{ W(C) : C \in \mathcal C, \out = \AND{}\}}$.\\ \label{line::compute_phi_start}
            ${O_{i, j, k} := \argmin \{ W(C) : C \in \mathcal C, \out = \OR{}\}}$.\\ \label{line::compute_phi_end}
         }
      }
    }
    $C := \argmin\{W(A_{0,0,m-1}), W(O_{0,0,m-1})\}$.\; \label{line::final_choice}
    \Return Circuit $C'$ resulting from applying \cref{lem-soft-huffman} to $C$.\; \label{line::final_huff}
   \caption{\aop{} optimization.}
   \label{alg::main_algorithm}
 \end{algorithm}

\begin{theorem}\label{theorem::our_approx_gurantee}
\cref{alg::main_algorithm} computes a circuit with delay at most
\[\log_2 (W) + \log_2 \log_2 (m) + \log_2 \log_2 \log_2 (m) + 4.3\]
and can be implemented to run in time $\mathcal O(m^4)$.
\begin{proof} {\sc (Sketch)}
\cref{alg::main_algorithm}
considers, in particular, all recursion steps
from \cite{BrennerHermann2019}.
Using this,
one can show that
for any sub-instance $\phi_{i, j, k}$,
the algorithm computes a solution
which is at least as good as the solution computed by the algorithm from \cite{BrennerHermann2019}
and thus also meets the delay bound \labelcref{our bound}.
The running time is dominated by $\mathcal O(m^4)$ calls to \cref{alg::merge_circuits},
which can be implemented to run in constant time if only weights and delays are computed
and only the final circuit $C'$
in \cref{line::final_huff} of \cref{alg::main_algorithm} is actually constructed.
\end{proof}
\end{theorem}

We conjecture that a stronger theoretical delay bound can be proven for our algorithm.

In \cref{sec:experiments},
we will see that in our practically applied logic optimization framework,
the running time of \cref{alg::main_algorithm} is negligible.

In order to take care of the circuit size,
we can modify \cref{alg::main_algorithm}  as follows:
For each sub-instance $\phi_{i, j, k}$, we
store not just one circuit with the best delay per output gate type, but all non-dominated
circuits. Here, circuit $C$ dominates circuit $C'$
if both weight and size of $C$ are at least as good as in $C'$
and if the gate types of $\generalout{}(C)$ and $\generalout{}(C')$ coincide.
In the end, we choose $C$ to be the smallest among all weight-optimum circuits.
This does not affect the
delay of the circuit (and \cref{theorem::our_approx_gurantee}
still holds), but often reduces its size.

\section{Logic Optimization Framework}\label{sec:flow}

\begin{figure*}[t]
\centering
\includegraphics[width=0.9\textwidth]{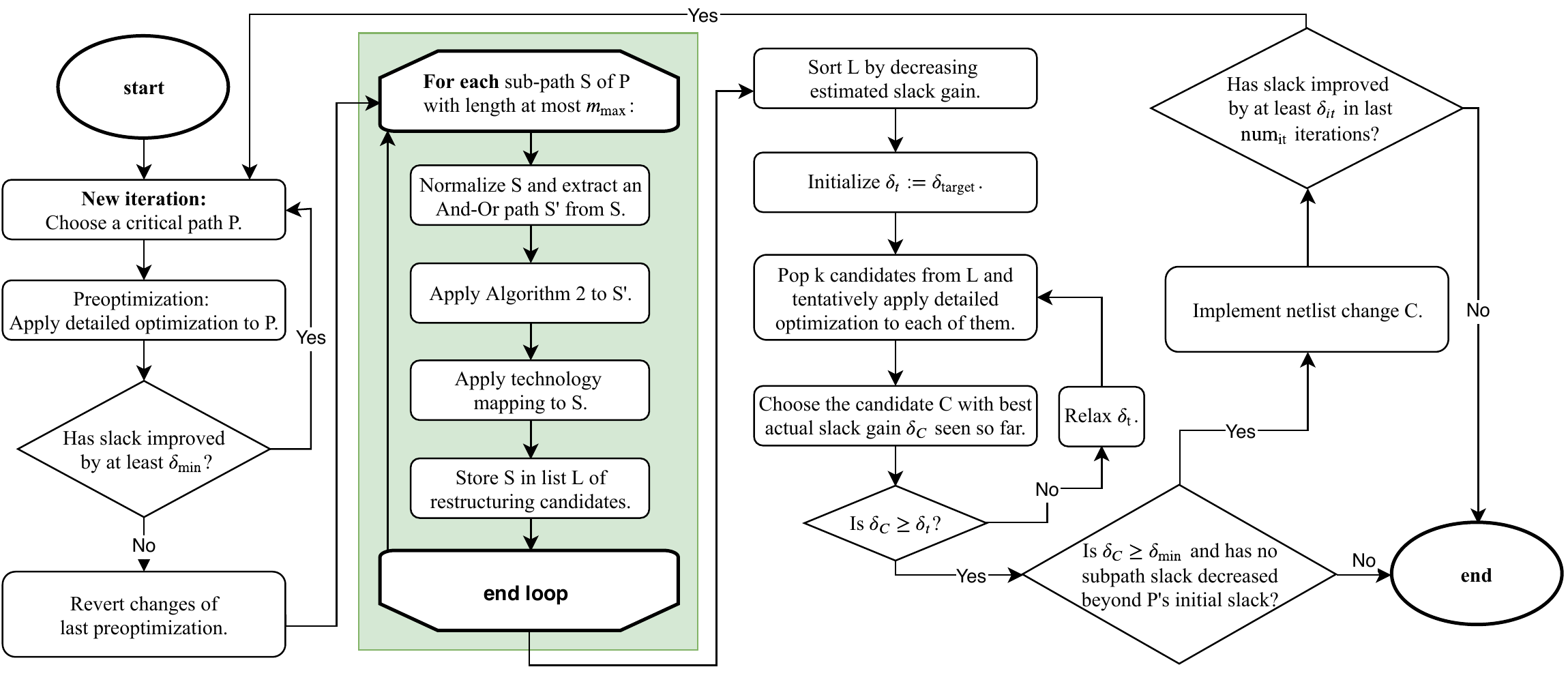}
\caption{Flow chart for our logic optimization framework (cf.~\cref{sec:flow}) with the path restructuring step in green.
 }
 \label{fig:flow_chart}
\end{figure*}

We propose a timing optimization framework (cf.~\cref{fig:flow_chart})
based on Werber et al.~\cite{WerberEtal2007}
with \cref{alg::main_algorithm}
as an essential component
that is used in production
in a late pre-routing stage of an industrial physical design flow.
Our framework revises the logical structure of critical paths
using placement and timing information.
In \cref{sec:normalization}, we adapt the delay model used in \cref{alg::main_algorithm}
to respect placement, buffering and gate sizing effects.
As we do not fully account for different kinds of gates or different gate sizes that might be available,
our framework involves a technology mapping step (\cref{sec:technology_mapping})
and powerful gate sizing and buffering routines (\cref{sec:detailed_opt}).

We iteratively optimize the worst slack of the currently most timing-critical combinational path
until overall worst slack does not improve significantly anymore.
A single \emph{iteration} works as follows:

Let $P$ denote a most critical path.
During a \emph{preoptimization} step,
we first try to improve the slack of $P$ without changing its logical structure
in order to diminish disruptions.
To this end, we apply \emph{detailed optimization} to $P$ as described in \cref{sec:detailed_opt}.
If a threshold slack improvement of $\deltamin$ is exceeded,
we keep the changes imposed by preoptimization and start the next iteration.

Otherwise, we discard the preoptimization's changes and perform the \emph{path restructuring} step
(central, green part of \cref{fig:flow_chart}).
This step works on internal data structures;
the netlist is not changed before detailed optimization (\cref{sec:detailed_opt}).
We consider the possibility to optimize any sub-path $S$ of $P$
up to a maximum length of $\maxpathlength$.
First, we apply a \emph{normalization} (\cref{sec:normalization})
in order to extract an \aop{} $S'$ from $S$
on which we run \cref{alg::main_algorithm}.
Then, the \emph{technology mapping} routine from \cite{Elbert2017}
(see also \cref{sec:technology_mapping}) locally modifies $S$
to benefit from all available gate types.
After having optimized all sub-paths of $P$,
we store all restructuring possibilities in a list $L$, sorted by decreasing estimated slack gain.

For only the most promising fraction of restructuring options,
we apply the time-consuming \emph{detailed optimization} (cf. \cref{sec:detailed_opt}).
First, we tentatively apply detailed optimization
to the topmost $k$ candidates in $L$.
If the actual slack gain of the best solution exceeds $\targetdelta$,
we choose this solution;
otherwise, we iteratively decrease $\targetdelta$ by a fixed value
and try out the next $k$ candidates in $L$
until we reach $\targetdelta$ or $L$ is empty.
Afterwards, we choose the restructuring candidate $C$ with best actual slack gain $\delta_C$ for $P$
among all detailed-optimized solutions.
This way, we usually apply detailed optimization to only a few instances,
but still find the overall best restructuring option.
If $\delta_C \geq \deltamin$ and if no side path slack has worsened beyond the initial slack of $P$,
we implement this netlist change,
possibly retaining parts of $P$ needed for side outputs.
If the change is implemented and the slack gain over the last $\itrange$ iterations exceeds a threshold $\deltarangeit$,
we start the next iteration;
otherwise, we stop.

Note that this is a simplified flow description.
E.g., in practice, we optimize the second critical path or the most critical latch-to-latch path
when $P$ cannot be further optimized.

\subsection{Normalization} \label{sec:normalization}

Our \aop{} optimization algorithm from \cref{sec:aop alg}
expects as an input an alternating path of \AND{}2 and \OR{}2 gates
with prescribed input arrival times,
and assumes that gates have a unit delay and connections do not impose any delay.
However, the most critical path $P$ contains arbitrary gates with varying delays,
and the physical locations of the path inputs might be far apart,
inducing undeniably high wire delays even after buffering.
A \emph{normalization} step thus transforms $P$ into a piece of netlist
whose core part is an \aop{} with appropriately modified input arrival times.

As we work on the most critical path,
the buffering routine applied in \cref{sec:detailed_opt}
will compute delay-optimum solutions.
Thus, we can assume a linear wire delay
and estimate the wire delay between two physical positions $p_1$ and $p_2$
by $\ddist \cdot ||p_1 - p_2||_1$ for a constant $\ddist \in \R$.
The traversal time through a gate is approximated by a constant $\dgate \in \R$.
The constants $\dgate$ and $\ddist$ are chosen based on an analysis of typical values
on the respective design.
As on the critical path, there are rather low fan-outs and slews,
the delay of gates with different types and sizes still varies,
but not much in comparison to the differences in arrival times.
Hence, assuming a realistic constant gate delay suffices
to determine the logical structure of the circuit.

Since we work on the most timing-critical part of the design,
we place the circuit $C$ computed by \cref{alg::main_algorithm}
such that each path is embedded delay-optimally,
implying that each path from an input $t_i$ to $\out$
has a wire delay of $\ddist \cdot ||l(t_i) - l(\out)||_1$,
where $l$ indicates physical coordinates on the chip.
Thus, the delay of $C$ is
$\max_{Q \colon t_i \rightsquigarrow \out} \big\{
   a(t_i) + \ddist \cdot ||l(t_i) - l(\out)||_1 + \dgate \cdot |Q| \big\},$
where the maximum ranges over all paths $Q$ in $C$ from any input $t_i$ to $\out$.
Applying \cref{alg::main_algorithm} with modified arrival times
\begin{equation*} \label{mod ats}
a'(t_i) := \frac{1}{\dgate} \Big(a(t_i) + \ddist \cdot ||l(t_i) - l(\out)||_1\Big)
\end{equation*}
hence yields a circuit with optimum wire delay with respect to physical locations.
In fact, we choose a placement that is netlength-optimum among all delay-optimum placements:
We determine $l(\out)$ based on its successors in the netlist
and place each gate at the median position of its predecessors and $\out$.

Now, we can describe our normalization.
Let $x$ denote the most critical input of a sub-path $S$ of $P$.
We represent each gate in $S$ using $\AND2$ and $\INV$ gates only.
This does not necessarily yield a path,
but we can recover the original critical path by following the signal flow of $x$,
obtaining a path $S'$.
By applying De Morgan transformations in reverse topological order,
we ensure that
$S'$ contains $\AND{}2$ and $\OR{}2$ gates only,
possibly adding inverters at the inputs of $S'$.
We use Huffman coding (\cref{theorem::huffman})
on chains of $\AND{}2$ gates (or $\OR{}2$ gates) in $S'$
to move less critical gates into $S \backslash S'$,
respecting physical locations by modifying arrival times as above.
This way, $S'$ becomes an \aop{} that -- with input arrival times $a'$ --
can be passed to \cref{alg::main_algorithm}.
\cref{fig::normalization} depicts the normalization on a path $S$ (left)
containing inverters (bubbles), \NOR{}, and \OAI{} gates.
On the right, we show $S$ after normalization with the \aop{} $S'$ colored.

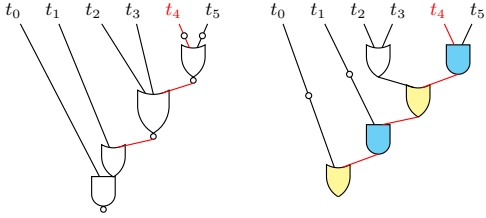
\begin{figure}

\newcommand{\subfigwidth}{0.36\columnwidth}
\newcommand{\picwidth}{0.97\columnwidth}
\newcommand{\atscale}{2}
\newcommand{\txtscale}{1.5}
\renewcommand{\showid}[2]{}
\newcommand{\showinputid}[2]{#1}

\centering

\adjustbox{valign=t}{
\begin{subfigure}[]{\subfigwidth}
\centering
\resizebox{\picwidth}{!}{%
\begin{tikzpicture}

\node[outer sep=0pt, scale=\txtscale] (i5) at (0.5, 3.2){$t_0$};
\node[outer sep=0pt, scale=\txtscale] (i4) at (1.5, 3.2){$t_1$};
\node[outer sep=0pt, scale=\txtscale] (i3) at (2.5, 3.2){$t_2$};
\node[outer sep=0pt, scale=\txtscale] (i2) at (3.5, 3.2){$t_3$};
\node[outer sep=0pt, scale=\txtscale, red] (i1) at (4.5, 3.2){$t_4$};
\node[outer sep=0pt, scale=\txtscale] (i0) at (5.5, 3.2){$t_5$};

\node[outer sep=0pt, nor gate US, draw, logic gate inputs=nn, rotate=270, thick] at (5,2) (nor1){};
\draw[thick] (i0) -- (nor1.input 1) node [midway, circle, draw=black, fill=white, scale=.5] {};
\draw[thick, red] (i1) -- (nor1.input 2) node [midway, circle, draw=black, fill=white, scale=.5] {};

\node[outer sep=0pt, nor gate US, draw, logic gate inputs=nnn, rotate=270, thick] at (4,0.75) (nor2){};
\draw[red, thick] (nor1.output) -- (nor2.input 1);
\draw[thick] (i2) -- (nor2.input 2);
\draw[thick] (i3) -- (nor2.input 3);

\node[outer sep=0pt, or gate US, draw, logic gate inputs=nn, rotate=270, thick] at (3,-0.5) (or2){};
\draw[thick, red] (nor2.output) -- (or2.input 1);
\draw[thick] (i4) -- (or2.input 2);

\node[outer sep=0pt, nand gate US, draw, logic gate inputs=nn, rotate=270, thick] at (2.75,-1.3) (nand5){};
\draw[thick] (i5) -- (nand5.input 2);
\end{tikzpicture}
}
\end{subfigure}
}
\adjustbox{valign=t}{
\begin{subfigure}[t]{\subfigwidth}
\centering
\resizebox{\picwidth}{!}{%
\begin{tikzpicture}

\node[outer sep=0pt, scale=\txtscale] (i5) at (0.5, 3.2){$t_0$};
\node[outer sep=0pt, scale=\txtscale] (i4) at (1.5, 3.2){$t_1$};
\node[outer sep=0pt, scale=\txtscale] (i3) at (2.5, 3.2){$t_2$};
\node[outer sep=0pt, scale=\txtscale] (i2) at (3.5, 3.2){$t_3$};
\node[outer sep=0pt, scale=\txtscale, red] (i1) at (4.5, 3.2){$t_4$};
\node[outer sep=0pt, scale=\txtscale] (i0) at (5.5, 3.2){$t_5$};

\node[fill=and_color, outer sep=0pt, and gate US, draw, logic gate inputs=nn, rotate=270, thick] at (5,2) (and1){};
\draw[thick] (i0) -- (and1.input 1);
\draw[thick, red] (i1) -- (and1.input 2);

\node[outer sep=0pt, or gate US, draw, logic gate inputs=nn, rotate=270, thick] at (3,2) (sideor){};
\draw[thick] (i2) -- (sideor.input 1);
\draw[thick] (i3) -- (sideor.input 2);

\node[fill=or_color, outer sep=0pt, or gate US, draw, logic gate inputs=nn, rotate=270, thick] at (4,1) (or1){};
\draw[red, thick] (and1.output) -- (or1.input 1);
\draw[thick] (sideor.output) -- (or1.input 2);

\node[fill=and_color, outer sep=0pt, and gate US, draw, logic gate inputs=nn, rotate=270, thick] at (3,0) (and4){};
\draw[thick, red] (or1.output) -- (and4.input 1);
\draw[thick] (i4) -- (and4.input 2) node [midway, circle, draw=black, fill=white, scale=.5] {};

\node[fill=or_color, outer sep=0pt, or gate US, draw, logic gate inputs=nn, rotate=270, thick] at (2,-1) (or2){};
\draw[thick, red] (and4.output) -- (or2.input 1);
\draw[thick] (i5) -- (or2.input 2) node [midway, circle, draw=black, fill=white, scale=.5] {};

\draw[thick, white] (or2.output) -- (2, -2);

\end{tikzpicture}
}
\end{subfigure}
}
\caption{A subpath $S$ of the critical path $P$ before (left)
and after normalization (right).
On the right, the extracted \aop{} $S'$ is colored.
Critical wires are drawn in red.}
\label{fig::normalization}
\end{figure}

\subsection{Technology Mapping} \label{sec:technology_mapping}

The purpose of our \emph{technology mapping} step is to change the newly created circuit locally
to improve worst slack and the physical area occupied by gates
by making use of all gates available on the design.
We use the dynamic programming algorithm from Elbert~\cite{Elbert2017}
which covers the input circuit by graphs representing the available gate types.
With respect to any fixed tradeoff of arrival time
(regarding our timing model from \cref{sec:normalization},
but with specific estimated delays per gate type)
and number of gates,
this algorithm computes an optimum technology mapping,
but the running time grows exponentially in the number
$l$ of gates with more than one successor.
In our application,
$l$ is usually very small, hence we can effort this running time
(cf. the end of \cref{sec:experiments}).
For constant $l$, \cite{Elbert2017}
also provides a fully polynomial-time approximation scheme.
On general circuits, computing a size-optimum technology mapping is NP-hard~\cite{Keutzer}.

\subsection{Detailed Optimization} \label{sec:detailed_opt}

Depending on the actual stage of the design,
our \emph{detailed optimization} step invokes buffering, layer assignment and gate sizing tools.
When used in late physical design,
we apply Held's gate sizing routine \cite{Held2009},
followed by the buffering tool with an integrated layer assignment
by Bartoschek et al.~\cite{Bartoschek2009}.
After buffering, we apply gate sizing again, in particular on newly inserted buffers.
As we work on the most critical fraction of the design,
$V_t$ assignment can be done conveniently by using the fastest gates available.

An incremental placement legalization makes sure that the placement remains legal throughout all netlist changes.

\section{Experimental Results}\label{sec:experiments}


In a first set of experiments, we examined the \aop{} optimization
algorithm from \cref{sec:aop alg} separately.
To this end, we created \aop{} instances with 4 to 28 inputs and
random integral arrival times chosen uniformly from the interval $[0, \# \text{inputs}]$.
For each number of inputs, we created 1000 instances.

We compared our results with the previously best methods
\cite{BrennerHermann2019}, \cite{HeldSpirkl2017}, and \cite{RautenbachEtal2006}.
For each instance, we ran all three algorithms
and compared the best result in terms of delay to our algorithm's output.
\cref{random-compare} visualizes our results.
Instances are grouped by their numbers of inputs,
and colors indicate the absolute delay difference of computed solutions.
Our algorithm covers all recursion options from \cite{HeldSpirkl2017,BrennerHermann2019,RautenbachEtal2006},
so our solutions can never be worse.
In fact, on almost all instances, the delay of our circuit is better,
and already for $18$ inputs, on every other instance better by $2$ or more.


\begin{figure}[bp]
\newcommand{\subfigwidth}{0.48\textwidth}
\begin{subfigure}[t]{\subfigwidth}
\centering
\includegraphics[width=0.96\textwidth]{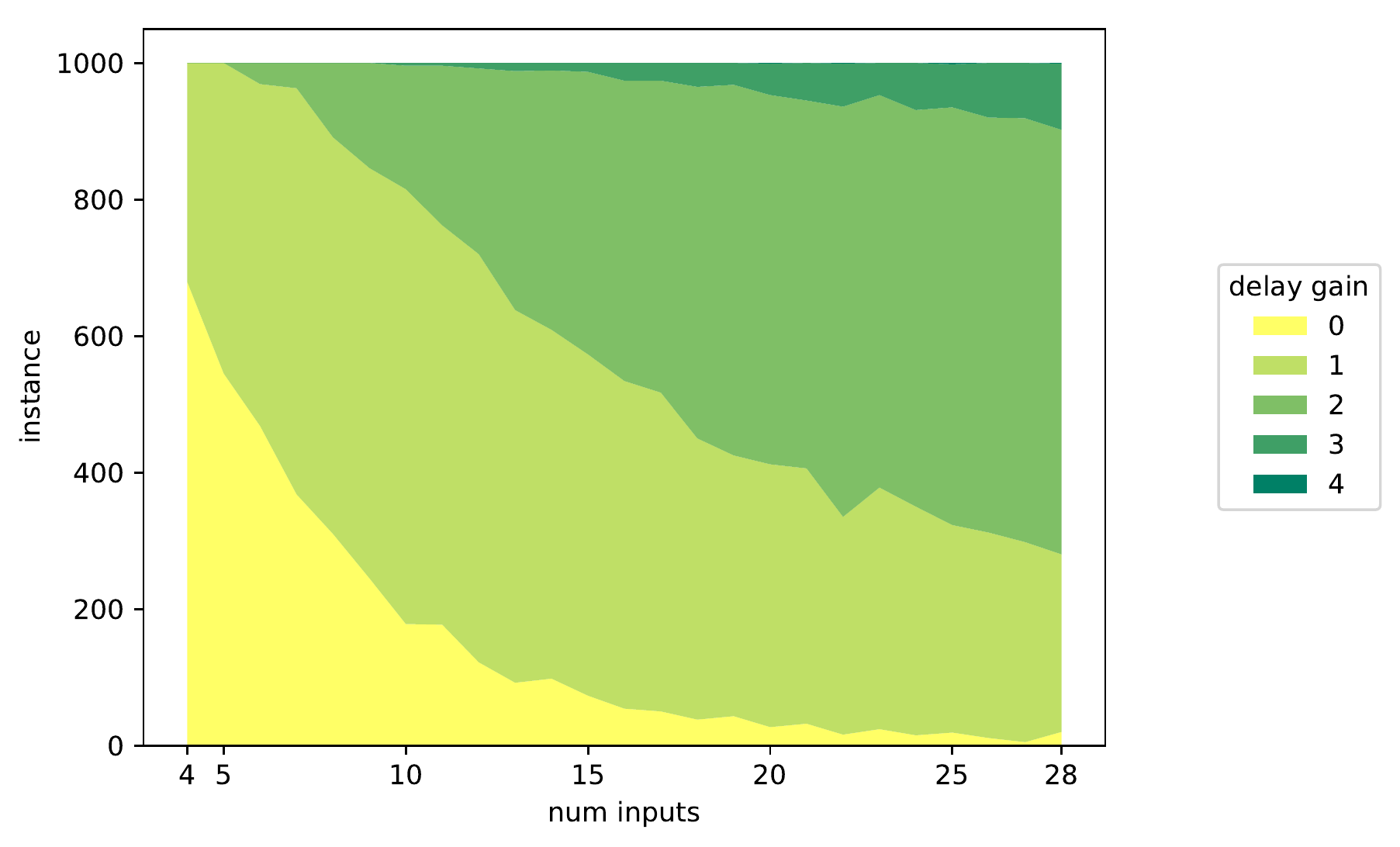}
\label{delay-comp}
\end{subfigure}
\caption{Delay gain of the solutions computed by \cref{alg::main_algorithm}
         compared to the best solution among \cite{HeldSpirkl2017,BrennerHermann2019,RautenbachEtal2006}
         on instances with random integral input arrival times.}
\label{random-compare}
\end{figure}

For each instance, we computed a lower bound on delay based
on the following ideas:
First, Kraft's inequality \cite{Kraft1949}
imposes a lower bound on the delay of any binary circuit;
secondly, we enumerate possible local gate configurations near the output
of an \aop{} circuit $C$ and recursively compute lower bounds for sub-circuits.
We compared our delay to the resulting lower bound.
Among all our solutions, \SI{89}{\percent} achieve the lower bound and hence are provably delay-optimum,
and only \SI{0.012}{\percent} exceed the lower bound by~$2$.

\cref{concrete-instance} compares our realization with \cite{HeldSpirkl2017} on an example instance.
In our circuit, the splits \labelcref{eq::rec_3}$^*$, \labelcref{eq::rec_2} and \labelcref{eq::rec_2}$^*$ were applied,
and the ability to optimize undetermined circuits was used twice.
This way, our delay of $22$ is better than the delay found by \cite{HeldSpirkl2017},
and it is even optimum since the input with arrival time $20$
has to traverse at least $2$ gates in any solution.
On this instance, we need one more gate than \cite{HeldSpirkl2017}.
In general, the number of gates used by our algorithm (with our modification for size reduction)
is typically higher than in
\cite{HeldSpirkl2017,BrennerHermann2019,RautenbachEtal2006},
but mostly in the range of \SI{20}{\percent}.

\input{instance_example.tex}

\begin{table}
\centering
\newcommand{\chipname}[2]{#2}
\renewcommand{\arraystretch}{0.8}
\setlength{\tabcolsep}{2pt}

\begin{tabular}{ llrrrrrrrrrrrrrr }
\toprule

Unit & Run & \multicolumn{1}{c}{WS [ps]} & \multicolumn{1}{c}{TS [ns]} & \multicolumn{1}{c}{\#\! Gates} & \multicolumn{1}{c}{Area} & \multicolumn{1}{c}{Netlength} & \multicolumn{1}{c}{ACE5} & \multicolumn{1}{c}{T [s]} \\
\midrule

%
%

\chipName{i1} & init & \rmUkn{201} & \rmUkn{15.3}  & \rmUkn{40636}  &                  &                   & \rmUkn[p]{85}  \\
              & LO   & \rmBtr{188} & \rmEql{15.3}  & \rmEql{40629}  & \rmEql[p]{-0.02} & \rmEql[p]{+0.00}  & \rmEql[p]{86} & \rmUkn{12}\\ \midrule
\chipName{i2} & init & \rmUkn{62}  & \rmUkn{52.2}  & \rmUkn{62185}  &                  &                   & \rmUkn[p]{96}  \\
              & LO   & \rmBtr{58}  & \rmEql{52.3}  & \rmEql{62187}  & \rmEql[p]{+0.02} & \rmEql[p]{+0.04}  & \rmEql[p]{96} & \rmUkn{11} \\ \midrule
\chipName{i3} & init & \rmUkn{109} & \rmUkn{192.9} & \rmUkn{69049}  &                  &                   & \rmUkn{107\%}  \\
              & LO   & \rmBtr{93}  & \rmBtr{189.4} & \rmEql{69066}  & \rmEql[p]{+0.01}& \rmEql[p]{+0.00}  & \rmEql{107\%}  & \rmUkn{273} \\ \midrule
\chipName{i4} & init & \rmUkn{5}   & \rmUkn{0.1}   & \rmUkn{78030}  &                  &                   & \rmUkn[p]{99}  \\
              & LO   & \rmBtr{0}   & \rmEql{0.0}   & \rmEql{77966}  & \rmEql[p]{-0.06} & \rmEql[p]{-0.07}  & \rmEql[p]{99} & \rmUkn{59} \\ \midrule
\chipName{i5} & init & \rmUkn{159} & \rmUkn{345.8} & \rmUkn{210828} &                  &                   & \rmUkn[p]{94}  \\
              & LO   & \rmBtr{152} & \rmBtr{343.4} & \rmEql{210852} & \rmEql[p]{+0.02} & \rmEql[p]{+0.00}  & \rmEql[p]{94} & \rmUkn{287} \\ \midrule
\chipName{i6} & init & \rmUkn{34}  & \rmUkn{13.0}  & \rmUkn{264744} &                  &                   & \rmUkn[p]{89}  \\
              & LO   & \rmBtr{20}  & \rmBtr{8.5}   & \rmEql{264724} & \rmEql[p]{+0.00} & \rmEql[p]{+0.01}  & \rmEql[p]{88} & \rmUkn{228} \\ \midrule
\chipName{i7} & init & \rmUkn{92}  & \rmUkn{251.5} & \rmUkn{272020} &                  &                   & \rmUkn[p]{96}  \\
              & LO   & \rmBtr{77}  & \rmBtr{230.1} & \rmEql{272242} & \rmEql[p]{+0.03} & \rmEql[p]{+0.06}  & \rmEql[p]{95} & \rmUkn{525} \\ \midrule
\chipName{i8} & init & \rmUkn{136} & \rmUkn{850.1} & \rmUkn{327807} &                  &                   & \rmUkn[p]{90}  \\
              & LO   & \rmBtr{120} & \rmBtr{833.1} & \rmEql{327916} & \rmEql[p]{+0.01} & \rmEql[p]{+0.02}  & \rmEql[p]{90} & \rmUkn{249} \\

%
%
\bottomrule
\end{tabular}
\caption{Performance of our logic restructuring framework on 7nm real-world instances.}
\label{table-elmore}
\end{table}

In a second set of experiments, we examined our logic optimization framework as a whole.
\cref{table-elmore} shows results on recent 7nm pre-routing designs using the RICE delay model.
The 'init' row displays the state of the chips as in our application
in industry:
a timing-driven placement has been computed, followed by various timing optimization steps,
among those our buffering and gate sizing sub-routines.
The initial netlist cannot be improved any further by classical timing optimization.
The '\BL' row shows results after applying our logic optimization flow to this netlist.
We see that worst slack (WS)
and the total sum of negative slacks (TS)
mostly improve significantly during logic optimization.
This does not disrupt global objectives as area, number of gates, netlength,
and routability, which barely change.
To check routability, we use the ACE5 estimate from \cite{GLARE},
the average congestion of the 5 \% most congested resources, weighted by usage,
computed by the global router from \cite{brg}.

Our program was implemented in C++,
and all tests were executed
on a machine with two Intel(R) Xeon(R) CPU E5-2667 v2 processors,
using a single thread.
In the last column (T), we show the total running time of our flow,
which is largely dominated by gate sizing
because it performs many expensive queries to the timing engine.
On any design, the total running time of all calls to \cref{alg::main_algorithm}
is less than $1$ second,
and less than $4$ seconds for the whole path restructuring step.
Per design, we consider roughly 1500 \aop{} restructuring instances with up to $13$ inputs.

\section{Conclusion}\label{sec:conclusion}

We presented a new approximation algorithm for delay optimization of \aop{}s
and a logic optimization framework using this algorithm to
improve critical paths in late physical design.
Regarding a simple, but realistic delay model,
our algorithm fulfills best known mathematical guarantees,
outperforms previously best approaches and is often optimum.
Results on industrial 7nm designs demonstrate that our logic optimization framework
improves timing when traditional timing optimization tools are at an end.

%
\bibliographystyle{plain}
\bibliography{extended}

%

\end{document}